\documentclass[pdflatex,sn-mathphys-num]{sn-jnl}

\usepackage{graphicx}%
\usepackage{multirow}%
\usepackage{amsmath,amssymb,amsfonts}%
\usepackage{amsthm}%
\usepackage{mathrsfs}%
\usepackage[title]{appendix}%
\usepackage{xcolor}%
\usepackage{textcomp}%
\usepackage{manyfoot}%
\usepackage{booktabs}%
\usepackage{algorithm}%
\usepackage{algorithmicx}%
\usepackage{algpseudocode}%
\usepackage{listings}%

\theoremstyle{thmstyleone}%
\newtheorem{theorem}{Theorem}
\newtheorem{proposition}{Proposition}%
\newtheorem{corollary}{Corollary}
\theoremstyle{thmstyletwo}%
\newtheorem{remark}{Remark}%

\theoremstyle{thmstylethree}%

\begin{document}

\title[Approximation, Consistency, and Prediction in GNMMs]
{Approximation, Consistency, and Prediction in Generalized Neural Network Mixed Models for Longitudinal and Clustered Data}

\author*[1]{\fnm{Mahmud} \sur{Hasan}}\email{mnh5524@psu.edu}

\author[2]{\fnm{Md Moinul} \sur{Ahsan}}\email{ahsanm8@vcu.edu}

\author[3]{\fnm{Mathias Nthiani} \sur{Muia}}\email{mnmuia@southalabama.edu}

\author[4]{\fnm{Mous-Abou} \sur{Hamadou}}\email{mhamadou@quinnipiac.edu}
\author[5]{\fnm{Chathurika S} \sur{Abeykoon}}\email{abeykoonc@rhodes.edu}
\small{\affil*[1]{\orgdiv{Department of Mathematics and Statistics}, \orgname{Pennsylvania State University-Altoona}, \orgaddress{\city{} PA, USA}}

\affil[2]{\orgdiv{Department of Biostatistics}, \orgname{Virginia Commonwealth University}, \orgaddress{\city{Richmond}, VA, USA}}

\affil[3]{\orgdiv{Department of Mathematics and Statistics}, \orgname{University of South Alabama}, \orgaddress{\city{Mobile}, AL, USA}}

\affil[4]{\orgdiv{Department of Mathematics and Statistics}, \orgname{Quinnipiac University}, \orgaddress{\city{Hamden}, CT, USA}}

\affil[5]{\orgdiv{Department of Mathematics and Statistics}, \orgname{Rhodes College}, \orgaddress{\city{Memphis}, TN, USA}}}

\abstract{
Generalized neural network mixed models (GNMMs) combine neural networks and
mixed-effects models to analyze longitudinal and clustered data. Although
GNMMs have shown good predictive performance in practice, their statistical
properties are not well understood. Therefore, this paper studies several
important properties of GNMM and provides theoretical support for its use.
First, we show that the neural-network component can approximate an unknown
nonlinear fixed-effect function arbitrarily well and establish consistency
of the fitted neural-network function for its population target within a
given network class. Next, we establish consistency of random-effect
estimation under increasing within-subject information and show that the
GNMM predictor converges to the true predictor. We then prove consistency
of the conditional mean function and show that the excess mean squared
prediction error (MSPE) relative to the oracle prediction error converges
to zero. We conduct simulation studies
to examine nonlinear fixed-effect recovery, random-effect estimation,
predictor and conditional mean estimation, and excess prediction error under
different numbers of subjects and repeated observations. We also apply GNMM
to a longitudinal schizophrenia dataset to predict continuous and binary
depression outcomes using smartphone-based behavioral measurements and
compare its predictive performance with generalized linear mixed models and
conventional artificial neural networks. Both the theoretical and empirical
results demonstrate that GNMM provides a flexible framework for modeling
nonlinear longitudinal and clustered data. These findings provide a
statistical foundation for GNMM prediction and support future methodological
developments.
}



\maketitle
\section{Introduction}\label{sec1}

Longitudinal and clustered data commonly arise in medicine, public health,
psychology, economics, and many other scientific fields. This type of data
involves repeated measurements collected from the same subjects over time.
For example, mobile health (mHealth) studies of individuals with schizophrenia
collect depression, anxiety, mobility, communication, and behavioral
measurements using smartphones over time \citep{Torous2016}. Because these
measurements are collected from the same individuals, the resulting
observations are correlated within subjects. Ignoring this within-subject
correlation may lead to inefficient estimation and poor prediction performance
\citep{diggle2002analysis,mcculloch2008generalized}.

Classical statistical methods such as linear mixed models (LMMs)
\citep{laird1982random} and generalized linear mixed models (GLMMs)
\citep{breslow1993approximate,mcculloch2008generalized} are widely used for
analyzing longitudinal and clustered data. GLMMs extend generalized linear
models by incorporating subject-specific random effects to account for
within-subject correlation and heterogeneity \citep{verbeke1996heterogeneity}.
They have been successfully applied to continuous, binary, count, and other
non-Gaussian outcomes. However, standard GLMMs assume a linear fixed-effect
component, which may not adequately capture complex nonlinear relationships
between predictors and outcomes.

Several extensions have been proposed to address this limitation. Nonlinear
mixed-effects (NLME) models \citep{lindstrom1990nonlinear} provide greater
flexibility for modeling nonlinear trends, but their estimation procedures can
be computationally demanding and may not scale efficiently to high-dimensional
settings \citep{wortwein2023nme}. Similarly, generalized additive mixed models
(GAMMs) \citep{lin1999gamm} use smoothing splines to model nonlinear effects,
but they may have difficulty capturing complex interactions among predictors
\citep{wood2017generalized}.

Beyond these conventional nonlinear modeling approaches, deep learning
methods, particularly artificial neural networks (ANNs), provide a powerful
framework for modeling complex nonlinear relationships
\citep{goodfellow2016deep,lecun2015deep}. Feed-forward ANNs are universal
function approximators and can automatically learn nonlinear patterns from
data without requiring explicit model specification
\citep{cybenko1989approximation,hornik1989multilayer}. However, conventional ANNs assume that observations are independent, making
them unsuitable for longitudinal data. Neural-network approaches have been
explored for longitudinal data \citep{tandon2006neural}, while early extensions
that explicitly incorporated subject-specific information into ANNs
\citep{maity2013subject} often suffered from scalability issues as the number
of subjects increased \citep{mandel2023neural}.

To address this limitation, \citep{mandel2023neural} proposed the generalized
neural network mixed model (GNMM), which incorporates deep learning into the
generalized linear mixed model framework. Specifically, GNMM replaces the
linear fixed-effect component of a GLMM with a neural network while retaining
subject-specific random effects. This formulation enables GNMM to capture
complex nonlinear fixed-effect relationships while accounting for
within-subject correlation and heterogeneity.

Related neural mixed-effects approaches have increasingly been developed for
longitudinal and clustered data. GNMM has been applied to depression and
anxiety prediction using smartphone-based measurements from individuals with
schizophrenia \citep{mandel2023neural}. Related approaches combine neural
networks and random effects for flexible modeling of clustered data and
personalized nonlinear prediction \citep{tran2017random,tran2020deep,wortwein2023nme}. More
recently, GNMM and related neural mixed-effects models have been investigated
for longitudinal Parkinson's disease progression \citep{tong2026parkinson},
while neural generalized mixed-effects models have further extended this
framework through flexible neural representations and approximate
marginal-likelihood estimation \citep{slavutsky2026ngmm}.

Although these developments demonstrate the practical usefulness of neural
mixed-effects models, the theoretical properties of GNMM remain only
partially developed. In particular, \citep{mandel2023neural} derived a bound
comparing the prediction error of GNMM and LMM when the true fixed-effect
structure is linear, but general consistency results for nonlinear
fixed-effect recovery and subject-specific prediction remain limited.
Specifically, it remains unclear under general nonlinear fixed-effect
structures whether the neural-network class can approximate the true
fixed-effect function, whether the fitted neural-network function consistently
estimates its population target, whether subject-specific random effects can
be consistently estimated under increasing within-subject information,
whether the fitted predictor and conditional mean converge to their true
counterparts, and whether the resulting prediction error approaches the
corresponding oracle prediction error. Addressing these questions is essential
for establishing a statistical foundation for GNMM.

The primary objective of this paper is to develop a statistical foundation
for generalized neural network mixed models. We establish the approximation
property of the neural-network fixed-effect function and distinguish the
model-class approximation error from the estimation error associated with
fitting the neural-network parameters. We establish fixed-effect estimation
consistency within a given neural-network class, random-effect consistency
under increasing within-subject information, predictor consistency,
conditional mean consistency, and prediction-error consistency relative to
an oracle predictor. These results provide theoretical justification for
applying GNMM to longitudinal and clustered prediction problems.

We further investigate the finite-sample performance of GNMM through
simulation studies for both continuous and binary responses. The simulations
examine nonlinear fixed-effect recovery, random-effect estimation, predictor
and conditional mean estimation, and prediction performance as the number of
subjects and repeated observations per subject increase. Mean squared errors
are used to assess recovery of the fixed-effect function, random effects,
predictor, and conditional mean, together with the correlation between the
estimated and true random effects. Prediction performance is assessed using
mean squared prediction error (MSPE) and mean absolute error (MAE), with the
area under the receiver operating characteristic curve (AUC) additionally
considered for binary responses. We also compare the GNMM and oracle MSPEs
through the excess prediction error, providing a direct finite-sample
assessment of prediction-error consistency.

We also demonstrate the practical application of GNMM using a longitudinal
schizophrenia dataset collected from a smartphone-based mobile health
(mHealth) study \citep{Torous2016}. The analysis considers both continuous
and binary depression outcomes and compares GNMM with generalized linear
mixed models (GLMMs) and conventional artificial neural networks (ANNs).
For the continuous outcome, out-of-sample prediction is evaluated using
MSPE, root mean squared error (RMSE), and MAE, while the binary outcome is
evaluated using AUC, accuracy, sensitivity, and specificity. Since the true
fixed-effect function, random effects, predictor, and conditional mean are
unknown in the real data, these measures assess practical predictive
performance rather than directly evaluating the theoretical consistency
results.

The main contributions of this paper are threefold. First, we establish a
statistical foundation for GNMM through neural-network approximation,
fixed-effect and random-effect consistency, predictor and conditional mean
consistency, and prediction-error consistency relative to an oracle predictor.
Second, we provide finite-sample evidence corresponding to these theoretical
properties through simulation studies evaluating fixed-effect and random-effect
recovery, predictor and conditional mean estimation, and excess prediction
error. Third, we illustrate the practical applicability of GNMM using a
longitudinal schizophrenia dataset and compare its predictive performance
with GLMM and ANN for continuous and binary outcomes.

The remainder of the paper is organized as follows. Section~\ref{sec2}
introduces related models, followed by the GNMM formulation and notation.
Section~\ref{sec3} develops a statistical foundation for GNMM through
approximation, consistency, and prediction results.
Section~\ref{sec:simulation} presents the simulation design and evaluation
criteria, followed by an analysis of the simulation results. Section~\ref{sec5}
presents the real-data analysis using the longitudinal schizophrenia dataset.
Section~\ref{sec6} concludes the paper with a discussion and future research
directions.
\section{Model and Notations}
\label{sec2}

Generalized linear mixed models (GLMMs) provide a common framework for
analyzing clustered and longitudinal data by combining fixed effects and
subject-specific random effects
\cite{breslow1993approximate,mcculloch2008generalized}. Let
\[
\mathcal{D}
=
\{(Y_{ij},X_{ij},Z_{ij}): i=1,\ldots,m,\; j=1,\ldots,n_i\}
\]
denote a clustered dataset. Here, \(m\) is the number of subjects or
clusters, \(n_i\) is the number of observations for subject \(i\), and
\(N=\sum_{i=1}^{m} n_i\) is the total sample size. For subject \(i\) and
observation \(j\), \(Y_{ij}\) is the response variable,
\(X_{ij}\in\mathbb{R}^{p}\) is the vector of fixed-effect covariates, and
\(Z_{ij}\in\mathbb{R}^{q}\) is the vector of random-effect covariates. The
subject-specific random effect is denoted by \(b_i\in\mathbb{R}^{q}\). Let
\[
\mu_{ij}
=
E(Y_{ij}\mid X_{ij},Z_{ij},b_i)
\]
be the conditional mean of \(Y_{ij}\). A generalized linear mixed model is
written as
\begin{align}
g(\mu_{ij})
=
X_{ij}^{T}\beta
+
Z_{ij}^{T}b_i,
\end{align}
where \(g(\cdot)\) is a known link function,
\(\beta\in\mathbb{R}^{p}\) is the fixed-effect coefficient vector, and
\(b_i\sim N(0,D)\), with \(D\) a positive-definite covariance matrix. Artificial neural networks (ANNs) have been widely used to model complex
nonlinear relationships between predictors and responses because of their
universal approximation capability
\citep{hornik1989multilayer,goodfellow2016deep}. A standard feed-forward ANN
models the conditional mean as
\[
g(\mu_{ij})
=
f_{\theta}(X_{ij}),
\]
where \(f_{\theta}(\cdot)\) is a neural network with parameter vector
\(\theta\). Although ANNs can effectively capture nonlinear
predictor--response relationships, the standard ANN formulation does not
explicitly account for within-subject dependence or subject-specific
heterogeneity commonly present in longitudinal and clustered data. To accommodate nonlinear relationships between the covariates and the
response, \citep{mandel2023neural} proposed the generalized neural network
mixed model (GNMM), which replaces the linear fixed-effect component
\(X_{ij}^{T}\beta\) in the GLMM with a feed-forward neural network.
Specifically,
\begin{align}
g(\mu_{ij})
=
f_{\theta}(X_{ij})
+
Z_{ij}^{T}b_i,
\end{align}
where \(f_{\theta}(\cdot)\) is a feed-forward neural network with parameter
vector \(\theta\), \(g(\cdot)\) is a known link function, and
\(b_i\sim N(0,D)\) denotes the subject-specific random effects. Suppose the true data-generating model is
\begin{align}
g(\mu_{ij}^{0})
=
f_{0}(X_{ij})
+
Z_{ij}^{T}b_i,
\end{align}
where
\[
\mu_{ij}^{0}
=
E(Y_{ij}\mid X_{ij},Z_{ij},b_i),
\]
\(f_{0}:\mathcal{X}\rightarrow\mathbb{R}\) is an unknown nonlinear
fixed-effect function, and \(b_i\sim N(0,D)\), with \(D\) denoting the
covariance matrix of the random effects. Let
\begin{align}\label{Tpr}
\eta_{ij}^{0}
=
f_{0}(X_{ij})
+
Z_{ij}^{T}b_i
\end{align}
denote the true linear predictor and
\begin{align}\label{GPr}
\widehat{\eta}_{ij}
=
f_{\widehat{\theta}}(X_{ij})
+
Z_{ij}^{T}\widehat{b}_i
\end{align}
denote the fitted GNMM predictor. Here, \(\widehat{\theta}\) is the estimator
of the neural-network parameter vector \(\theta\) obtained from the GNMM
fitting procedure based on quasi-likelihood estimation
\citep{mandel2023neural}, and \(\widehat{b}_i\) is the estimated
subject-specific random effect for subject \(i\). Let
\[
\widehat{\mu}_{ij}
=
g^{-1}(\widehat{\eta}_{ij})
\]
denote the fitted conditional mean and
\[
\mu_{ij}^{0}
=
g^{-1}(\eta_{ij}^{0})
\]
denote the true conditional mean. Furthermore, let \(\|\cdot\|\) denote the
Euclidean norm and let \(o_p(1)\) denote a sequence of random quantities
converging to zero in probability under the asymptotic regime specified for
each result.

For \(K\geq1\), let
\[
\mathcal{F}_{K}
=
\{f_{\theta}:\theta\in\Theta_{K}\}
\]
denote a class of feed-forward neural networks, where \(K\) indexes the
network complexity and \(\Theta_K\) denotes the corresponding parameter
space. We consider a sequence of neural-network classes
\(\{\mathcal{F}_{K}\}_{K\geq1}\) satisfying the universal approximation
property
\citep{cybenko1989approximation,hornik1989multilayer}. Thus,
\(f_{\theta_K}\in\mathcal F_K\) denotes an approximating neural-network
function within a network class of complexity \(K\), whereas
\(f_{\widehat{\theta}}\) denotes the neural-network function estimated from
the observed data. This distinction separates neural-network approximation
from statistical estimation in the theoretical development that follows. We assume that the covariate space
\(\mathcal{X}\subset\mathbb{R}^{p}\) is compact, the function \(f_{0}\) is
continuous on \(\mathcal{X}\), the covariance matrix \(D\) is positive
definite, and
\[
\operatorname{Var}
(Y_{ij}\mid X_{ij},Z_{ij},b_i)
<
\infty.
\]
Additional regularity conditions required for individual consistency and
prediction-error results are stated with the corresponding results in
Section~\ref{sec3}.

\section{Statistical Properties of Generalized Neural Network Mixed Models}
\label{sec3}

Under the assumptions stated in Section~\ref{sec2}, we investigate the
statistical properties of GNMM by separating the main sources of error in
the fitted predictor. First, we study the approximation error of the
neural-network component and show that a sufficiently rich neural-network
class can approximate the unknown nonlinear fixed-effect function \(f_0\).
Second, we distinguish this approximation error from the estimation error
arising from fitting the neural-network parameters. Third, we study the
consistency of the estimated subject-specific random effects
\(\widehat b_i\) under increasing within-subject information. We then combine
these results to decompose the difference between the fitted predictor
\(\widehat{\eta}_{ij}\) and the true predictor \(\eta_{ij}^{0}\) into
neural-network estimation error, approximation error, and random-effect
estimation error. This decomposition leads to consistency of the fitted
predictor and conditional mean and provides an explicit characterization of
the excess prediction error relative to the oracle predictor. Together,
these results provide a statistical foundation for GNMM prediction in
longitudinal and clustered data.
\subsection{Approximation Capability of GNMM}\label{appr}

The approximation property establishes that the neural-network function
\(f_{\theta_K}\) of GNMM can approximate the unknown nonlinear fixed-effect
function \(f_{0}\), as shown in Theorem~\ref{Thm1}. The theorem applies the
classical universal approximation property to the GNMM framework by showing
that replacing the linear fixed-effect component with
\(f_{\theta_K}(X_{ij})\) enables the resulting GNMM predictor
\(\eta_{ij}^{K}\) to approximate the true predictor
\(\eta_{ij}^{0}\) while preserving the mixed-effects structure. This provides
the theoretical justification for using GNMM to model complex nonlinear
relationships in longitudinal and clustered data. This approximation result
concerns the representational capacity of the neural-network class; the
estimation error associated with the fitted neural-network function
\(f_{\widehat{\theta}}\) is considered separately in the subsequent results.

\begin{theorem}[Approximation Property]\label{Thm1}

Let \(\mathcal{X}\subset\mathbb{R}^{p}\) be compact and let
\(f_{0}:\mathcal{X}\rightarrow\mathbb{R}\) be continuous. Suppose that
the sequence of neural-network classes
\(\{\mathcal{F}_{K}\}_{K\geq 1}\) satisfies the universal approximation
property, where \(K\) indexes the network complexity. Then, for every
\(\varepsilon>0\), there exists a sufficiently large \(K\) and a
neural-network function \(f_{\theta_{K}}\in\mathcal{F}_{K}\) such that

\[
\sup_{i,j}
\left|
\eta_{ij}^{K}
-
\eta_{ij}^{0}
\right|
<
\varepsilon,
\]

where \(\eta_{ij}^{0}\) is the true linear predictor defined in
\eqref{Tpr}, and
\[
\eta_{ij}^{K}
=
f_{\theta_K}(X_{ij})
+
Z_{ij}^{T}b_i
\]
is the corresponding GNMM predictor obtained using
\(f_{\theta_K}\) and the same subject-specific random effect \(b_i\)
as in the true predictor.

\end{theorem}

Hence, the GNMM predictor can approximate the true predictor arbitrarily
well as the neural-network class becomes sufficiently rich.

\begin{remark}
Theorem~\ref{Thm1} applies the classical universal approximation result to
the generalized neural network mixed-model framework. It shows that replacing
the linear fixed-effect component with a neural network preserves
approximation of the complete linear predictor while retaining the
mixed-effects structure. Since the same random-effect term
\(Z_{ij}^{T}b_i\) appears in both \(\eta_{ij}^{K}\) and
\(\eta_{ij}^{0}\), the approximation error in this result arises entirely
from approximating the unknown fixed-effect function \(f_0\) by
\(f_{\theta_K}\). The result therefore concerns the representational
capacity of the neural-network class and does not by itself establish
consistency of the estimated neural-network function
\(f_{\widehat{\theta}}\). The distinction between approximation error and
estimation error is considered separately in the subsequent development of
predictor consistency. In the simulation study, this result motivates the
use of nonlinear data-generating mechanisms, while in the real-data
application it provides theoretical motivation for allowing potentially
complex nonlinear predictor--response relationships.
\end{remark}
\subsection{Consistency of the Neural-Network Fixed-Effect Estimation}
\label{consistency of f}

Theorem~\ref{Thm1} establishes the approximation capability of the
neural-network class, but it does not by itself establish consistency of the
fitted neural-network function \(f_{\widehat{\theta}}\). We therefore
distinguish the approximation error from the estimation error associated
with fitting the neural-network parameters. For a given network class
\(\mathcal F_K\), let \(\theta_K^{*}\) denote the population target of the
GNMM estimation criterion, and let \(f_{\widehat{\theta}}\) denote the
neural-network function estimated from the observed data.
\begin{proposition}[Fixed-Effect Estimation Consistency]
\label{PropFixed}

For a fixed network class \(\mathcal F_K\), suppose that the parameter space
\(\Theta_K\) is compact, the population quasi-likelihood criterion is
continuous and uniquely maximized at \(\theta_K^{*}\), and the corresponding
sample quasi-likelihood criterion converges uniformly in probability to its
population counterpart. Then
\[
\widehat{\theta}
\xrightarrow{p}
\theta_K^{*}.
\]
If, in addition, the mapping
\(\theta\mapsto f_{\theta}\) is continuous with respect to the uniform norm
on \(\mathcal X\), then
\[
\sup_{x\in\mathcal X}
\left|
f_{\widehat{\theta}}(x)
-
f_{\theta_K^{*}}(x)
\right|
=
o_p(1).
\]
Thus, for a fixed network class, the fitted neural-network function is
consistent for its population target \(f_{\theta_K^{*}}\).

\end{proposition}
\begin{remark}
Proposition~\ref{PropFixed} concerns estimation within a fixed
neural-network class and is distinct from the approximation property in
Theorem~\ref{Thm1}. Theorem~\ref{Thm1} characterizes the ability of a
sufficiently rich neural-network class to approximate \(f_0\), whereas
Proposition~\ref{PropFixed} establishes conditions under which the fitted
neural-network function approaches its population target
\(f_{\theta_K^{*}}\) within a given class. The population target
\(f_{\theta_K^{*}}\) need not coincide with the existence-based
approximating function \(f_{\theta_K}\) in Theorem~\ref{Thm1}; therefore,
control of the discrepancy between \(f_{\theta_K^{*}}\) and \(f_0\) is
imposed separately when establishing predictor consistency. Thus,
approximation error and estimation error are treated as separate components
of the GNMM prediction error in the subsequent analysis.
\end{remark}
\subsection{Consistency of Random Effect Estimation}
\label{consistency of bi}

Although Theorem~\ref{Thm1} establishes the ability of GNMM to approximate
the underlying nonlinear fixed-effect function, accurate prediction for
longitudinal and clustered data also requires proper estimation of the
subject-specific random effects. Since the random effects capture
heterogeneity among subjects and induce within-subject dependence, their
consistent estimation is essential for reliable subject-level inference and
prediction. The following Proposition~\ref{Prop1} establishes the consistency
of the estimated random effects. Let \(\widehat{b}_i\) denote the estimator
of the random effect \(b_i\) obtained from the fitted GNMM, and let
\(\ell_i(b)\) denote the corresponding subject-specific conditional
log-likelihood for subject \(i\). For the random-effect consistency result,
we assume consistency of the estimated fixed-effect function for \(f_0\).

\begin{proposition}[Random Effect Consistency]\label{Prop1}

Assume that the GNMM is correctly specified and that
\[
\sup_{x\in\mathcal X}
\left|
f_{\widehat{\theta}}(x)-f_0(x)
\right|
=
o_p(1).
\]
Assume further that the covariance matrix \(D\) is positive definite and the
covariates \(X_{ij}\) and \(Z_{ij}\) are bounded. Suppose that the normalized
subject-specific conditional log-likelihood converges uniformly in
probability to a continuous limiting objective function that is uniquely
maximized at \(b_i\). If \(n_i\rightarrow\infty\), then
\[
\widehat b_i\xrightarrow{p}b_i.
\]

\end{proposition}

\begin{remark}
Proposition~\ref{Prop1} shows that the estimated random effects converge to
their true values as the amount of information within each subject increases.
The result therefore applies to an increasing within-subject asymptotic
regime in which \(n_i\rightarrow\infty\); it does not claim consistent
recovery of an individual random effect when the number of repeated
observations per subject remains bounded. In simulation studies, this
property can be examined by comparing the estimated and true random effects
across different cluster sizes. In practical applications, it provides
theoretical support for subject-specific prediction when sufficient
within-subject information is available.
\end{remark}
\subsection{Consistency of the GNMM Predictor}
\label{Consistency Predictor}

Theorem~\ref{Thm1} establishes the approximation capability of the
neural-network class, while Proposition~\ref{PropFixed} establishes
consistency of the fitted neural-network function for its population target.
Together with control of the discrepancy between the population target and
\(f_0\), these results can be combined with the random-effect consistency
result in Proposition~\ref{Prop1} to study consistency of the fitted GNMM
predictor. The following Theorem~\ref{Thm2} shows that the estimated
predictor converges to the true predictor and explicitly identifies the
sources of prediction error.

\begin{theorem}[Predictor Consistency]\label{Thm2}

Suppose that the conditions of Proposition~\ref{PropFixed} hold,
\[
\widehat b_i
\xrightarrow{p}
b_i,
\]
as established in Proposition~\ref{Prop1}, and
\[
\sup_{i,j}\|Z_{ij}\|
\leq
C_Z<\infty.
\]
Suppose further that the population target \(f_{\theta_K^{*}}\) satisfies
\[
\sup_{x\in\mathcal X}
\left|
f_{\theta_K^{*}}(x)-f_0(x)
\right|
\longrightarrow 0.
\]
Let \(\eta_{ij}^{0}\) and \(\widehat{\eta}_{ij}\) denote the true and fitted
linear predictors defined in \eqref{Tpr} and \eqref{GPr}, respectively.
Then
\[
\begin{aligned}
\widehat{\eta}_{ij}-\eta_{ij}^{0}
&=
\underbrace{
f_{\widehat{\theta}}(X_{ij})
-
f_{\theta_K^{*}}(X_{ij})
}_{\text{neural-network estimation error}}
\\
&\quad+
\underbrace{
f_{\theta_K^{*}}(X_{ij})
-
f_0(X_{ij})
}_{\text{model-class approximation error}}
\\
&\quad+
\underbrace{
Z_{ij}^{T}(\widehat b_i-b_i)
}_{\text{random-effect estimation error}}.
\end{aligned}
\]

Consequently,
\begin{align}\label{CP}
\widehat{\eta}_{ij}
-
\eta_{ij}^{0}
=
o_p(1).
\end{align}

Equivalently,
\[
\widehat{\eta}_{ij}
\xrightarrow{p}
\eta_{ij}^{0}.
\]

\end{theorem}

\begin{remark}
Theorem~\ref{Thm2} establishes predictor consistency while explicitly
separating the three sources of error in the fitted GNMM predictor:
neural-network estimation error, model-class approximation error, and
random-effect estimation error. Thus, predictor consistency follows when
the fitted neural-network function approaches its population target
\(f_{\theta_K^{*}}\), the discrepancy between \(f_{\theta_K^{*}}\) and
\(f_0\) tends to zero, and the estimated subject-specific random effects
approach their true values. The population target
\(f_{\theta_K^{*}}\) need not coincide with the existence-based
approximating function \(f_{\theta_K}\) in Theorem~\ref{Thm1}.
This decomposition therefore separates estimation within the fitted
network class from the discrepancy between the population target and the
true fixed-effect function. It also provides theoretical support for
studying these error components separately in the simulation study.
\end{remark}
\subsection{Consistency of the Conditional Mean Function}
\label{Consistency mean}

While Theorem~\ref{Thm2} establishes consistency of the linear predictor,
practical prediction is typically based on the conditional mean obtained
through the inverse link function. Therefore, it is important to verify that
consistency of the predictor also translates into consistency of the
estimated conditional mean. Define the estimated conditional mean by
\[
\widehat{\mu}_{ij}
=
g^{-1}(\widehat{\eta}_{ij}),
\]
and the true conditional mean by
\[
\mu_{ij}^{0}
=
g^{-1}(\eta_{ij}^{0}),
\]
where \(g^{-1}(\cdot)\) denotes the inverse link function. The following
Corollary~\ref{Cor1} establishes this result.

\begin{corollary}\label{Cor1}

Suppose the conditions of Theorem~\ref{Thm2} hold and the inverse link
function \(g^{-1}\) is continuous at \(\eta_{ij}^{0}\). Then
\[
\widehat{\mu}_{ij}
-
\mu_{ij}^{0}
=
o_p(1),
\]
or equivalently,
\[
\widehat{\mu}_{ij}
\xrightarrow{p}
\mu_{ij}^{0}.
\]

\end{corollary}

Hence, the estimated conditional mean produced by GNMM is a consistent
estimator of the true conditional mean.

\begin{remark}
Corollary~\ref{Cor1} shows that consistency of the linear predictor carries
over to the conditional mean function through the inverse link. Since
prediction in generalized mixed models is commonly based on the conditional
mean, this result provides theoretical justification for the use of GNMM
predictions. The continuity condition is sufficient for convergence in
probability; stronger moment conditions are introduced separately when
studying mean squared prediction error in the next theorem.
\end{remark}
\subsection{Prediction Error Consistency}\label{Error consistency}

Corollary~\ref{Cor1} establishes consistency of the estimated conditional
mean. A natural next question is whether this consistency translates into
asymptotically optimal predictive performance under squared-error loss. To
answer this question, we compare the prediction error of GNMM with the oracle
prediction error obtained when the true conditional mean is known. The
following Theorem~\ref{Thm3} shows that the excess prediction error converges
to zero under mean-square control of the neural-network estimation error,
model-class approximation error, and random-effect estimation error.

Let
\begin{align}
MSPE_{\mathrm{GNMM}}
&=
E\left(
Y_{ij}
-
\widehat{\mu}_{ij}
\right)^2,
\label{eq:MSPE_GNMM}
\end{align}
denote the mean squared prediction error (MSPE) of GNMM, where \(Y_{ij}\)
denotes a new response from the same data-generating mechanism and
\(\widehat{\mu}_{ij}\) denotes the fitted conditional mean obtained from the
GNMM using the training data. We assume that
\[
E(Y_{ij}\mid X_{ij},Z_{ij},b_i)
=
\mu_{ij}^{0},
\]
and that, conditional on \(X_{ij}\), \(Z_{ij}\), and \(b_i\), the new
response \(Y_{ij}\) is independent of the training data used to construct
\(\widehat{\mu}_{ij}\). Similarly, let
\begin{align}
MSPE_{\mathrm{Oracle}}
&=
E\left(
Y_{ij}
-
\mu_{ij}^{0}
\right)^2,
\label{eq:MSPE_Oracle}
\end{align}
denote the oracle mean squared prediction error, where \(\mu_{ij}^{0}\)
denotes the true conditional mean. The expectation is taken with respect to
the joint distribution of the new response, covariates, random effects, and
the training data through \(\widehat{\mu}_{ij}\). The quantity
\(MSPE_{\mathrm{Oracle}}\) represents the minimum achievable mean squared
prediction error under squared-error loss when the true conditional mean is
known. The difference between \eqref{eq:MSPE_GNMM} and
\eqref{eq:MSPE_Oracle} measures the excess prediction error incurred by
estimating the conditional mean from the data.

\begin{theorem}[Prediction Error Consistency]\label{Thm3}

Suppose the conditions of Theorem~\ref{Thm2} and
Corollary~\ref{Cor1} hold, and assume that the inverse link function
\(g^{-1}\) is Lipschitz continuous with constant \(L_g<\infty\). If
\[
E\left[
\sup_{x\in\mathcal X}
\left|
f_{\widehat{\theta}}(x)
-
f_{\theta_K^{*}}(x)
\right|^2
\right]
\longrightarrow 0,
\]
\[
\sup_{x\in\mathcal X}
\left|
f_{\theta_K^{*}}(x)
-
f_0(x)
\right|^2
\longrightarrow 0,
\]
and
\[
E\left[
\|\widehat b_i-b_i\|^2
\right]
\longrightarrow 0,
\]
then
\[
MSPE_{\mathrm{GNMM}}
-
MSPE_{\mathrm{Oracle}}
\longrightarrow 0.
\]

\end{theorem}

\begin{remark}
Theorem~\ref{Thm3} shows that the prediction performance of GNMM approaches
that of the oracle predictor when the neural-network estimation error,
model-class approximation error, and random-effect estimation error vanish
in the required mean-square sense. Specifically, the fitted neural-network
function must approach its population target \(f_{\theta_K^{*}}\), the
discrepancy between \(f_{\theta_K^{*}}\) and \(f_0\) must tend to zero, and
the estimated subject-specific random effects must approach their true values
in mean square. These mean-square conditions are stronger than the
convergence-in-probability conditions used in the preceding consistency
results and are required to control mean squared prediction error. This
result provides a direct connection between the approximation, estimation,
and prediction properties developed in this section and gives theoretical
justification for evaluating GNMM using MSPE in the simulation study.
\end{remark}

\section{Simulation Study}\label{sec:simulation}

We conducted simulation studies to investigate the finite-sample behavior of
the proposed generalized neural network mixed model (GNMM). The simulations
were designed to provide finite-sample evidence corresponding to the
approximation and fixed-effect estimation properties, random-effect
consistency, predictor consistency, conditional mean consistency, and
prediction-error consistency results established in Section~\ref{sec3}.
The simulations are intended to examine whether the finite-sample behavior
of GNMM is consistent with the theoretical results rather than to establish
the asymptotic results themselves.

\subsection{Simulation Design}

For subject $i=1,\ldots,m$ and observation $j=1,\ldots,n_i$, the covariate
vector
$X_{ij}=(X_{ij1},X_{ij2},X_{ij3},X_{ij4})^{T}$ was generated independently
from a Uniform$(0,1)$ distribution. Subject-specific random effects were
generated from $b_i\sim N(0,\tau^2)$ with $\tau^2=0.25$. We considered the
following combinations of the number of subjects and repeated observations
per subject,
\[
(m,n_i)\in
\{(50,5),(50,10),(50,20),(100,10),(200,10),(200,20)\}.
\]
These settings were chosen to investigate the finite-sample behavior of
GNMM as the number of subjects $m$ and the number of repeated observations
$n_i$ increase. In particular, the settings with fixed $m=50$ and increasing
$n_i$ were included to examine the effect of increasing within-subject
information, corresponding to the asymptotic regime $n_i\rightarrow\infty$
underlying the random-effect consistency result in
Proposition~\ref{Prop1}. The settings with fixed $n_i=10$ and increasing
$m$ were used to examine the effect of increasing the number of subjects.
Thus, these settings allow the effects of increasing within-subject information and increasing the number of subjects to be examined separately. The data were generated according to the GNMM model
\begin{align}
g(\mu_{ij})
=
f_0(X_{ij})
+
b_i,
\label{eq:sim_model}
\end{align}

where $g(\cdot)$ denotes the appropriate link function and
\[
f_0(X_{ij})
=
\sin(2\pi X_{ij1})
+
X_{ij2}^{2}
+
X_{ij3}X_{ij4}.
\]
This nonlinear fixed-effect function was selected to evaluate the ability
of the neural-network component to recover the underlying nonlinear
fixed-effect structure considered in Theorem~\ref{Thm1} and
Proposition~\ref{PropFixed}. In particular, the neural-network component is
expected to approximate the true nonlinear function $f_0(\cdot)$. Since
$MSE(\widehat f,f_0)$ reflects both neural-network estimation error and
model-class approximation error, the simulation assesses overall
fixed-effect recovery rather than attempting to isolate these two error components separately. We considered two response types. Continuous responses were generated according to
\begin{align}
Y_{ij}
=
f_0(X_{ij})
+
b_i
+
\epsilon_{ij},
\qquad
\epsilon_{ij}\sim N(0,1),
\label{eq:sim_cont}
\end{align}
and binary responses were generated using the logistic model
\begin{align}
Y_{ij}
\sim
\mathrm{Bernoulli}(\mu_{ij}),
\qquad
\mathrm{logit}(\mu_{ij})
=
f_0(X_{ij})
+
b_i.
\label{eq:sim_bin}
\end{align}
We generated $R=500$ independent datasets for each simulation setting. Within each dataset, observations were randomly divided into 70\% training and 30\% testing sets within each subject, so that every subject contributed observations to both sets. The proposed GNMM was fitted to the training data using a feed-forward neural network with subject-specific random intercepts. The subject-specific random effects estimated from the training observations were then used to predict the corresponding testing observations. This within-subject splitting scheme evaluates prediction for additional observations from subjects represented
in the training data and is consistent with the subject-specific prediction setting considered in the theoretical development. Thus, the simulation focuses on prediction for existing subjects with estimated random effects rather than prediction for entirely new subjects. The feed-forward neural network consisted of one hidden layer with
8 hidden units and a ReLU activation function \citep{nair2010rectified}. Model estimation
was performed using BFGS optimization, with a maximum of 300 iterations, a relative convergence tolerance of $10^{-6}$, and an $L_2$ penalty with tuning parameter $\lambda=0.001$ applied to the neural-network parameters. The subject-specific random effects were penalized according to the
Gaussian random-effect distribution $b_i\sim N(0,\tau^2)$ with
$\tau^2=0.25$. The same neural-network architecture and
fitting procedure were used across all simulation settings to provide a
common basis for comparing finite-sample performance. Therefore, the
simulation evaluates finite-sample recovery under a fixed practical network
architecture; it does not directly reproduce the increasing-network-
complexity asymptotic regime used to make the model-class approximation
error vanish in the theoretical development.
\subsection{Evaluation Measures}

We evaluated the fitted GNMM using several performance measures corresponding
to the theoretical results established in Section~\ref{sec3}. Each measure
was first calculated separately for each of the $R=500$ simulation
replications and was then summarized across replications. In addition to the
Monte Carlo mean, the corresponding Monte Carlo standard error (MCSE) was
computed to quantify simulation variability. Specifically, if $Q_r$ denotes
the value of a given performance measure in replication $r$, $r=1,\ldots,R$,
then its Monte Carlo mean and MCSE were calculated as
\[
\overline{Q}
=
\frac{1}{R}\sum_{r=1}^{R}Q_r,
\qquad
MCSE(\overline{Q})
=
\frac{s_Q}{\sqrt{R}},
\]
where $s_Q$ denotes the sample standard deviation of
$Q_1,\ldots,Q_R$ across the simulation replications.

Let $\mathcal D_{\mathrm{test}}$ denote the set of held-out testing
observations and let $N_{\mathrm{test}}=|\mathcal D_{\mathrm{test}}|$ denote
the number of observations in the testing set. We first evaluated recovery
of the nonlinear fixed-effect function using
\begin{align}
MSE(\widehat f,f_0)
=
\frac{1}{N_{\mathrm{test}}}
\sum_{(i,j)\in\mathcal D_{\mathrm{test}}}
\left[
f_{\widehat{\theta}}(X_{ij})
-
f_0(X_{ij})
\right]^2,
\label{eq:mse_f}
\end{align}
where $f_{\widehat{\theta}}(\cdot)$ denotes the neural-network fixed-effect
function estimated from the training data and $f_0(\cdot)$ denotes the true
nonlinear fixed-effect function used to generate the data. This measure
quantifies the overall discrepancy between the fitted neural-network
fixed-effect function and the true nonlinear function $f_0$. Since this
discrepancy reflects both estimation within the fitted network class and the
discrepancy between the population target and $f_0$, it provides a
finite-sample assessment corresponding to the approximation and fixed-effect
estimation results in Theorem~\ref{Thm1} and
Proposition~\ref{PropFixed}. It should therefore be interpreted as a measure
of overall fixed-effect recovery rather than as a direct estimate of the
model-class approximation error alone.

We assessed the random-effect consistency result established in
Section~\ref{consistency of bi} using
$\mathrm{cor}(\widehat b_i,b_i)$ and
\begin{align}
MSE(\widehat b_i,b_i)
=
\frac{1}{m}
\sum_{i=1}^{m}
(\widehat b_i-b_i)^2,
\label{eq:mse_b}
\end{align}
where $b_i$ denotes the true random effect generated for subject $i$ and
$\widehat b_i$ denotes its estimate obtained from the GNMM fitted to the
training observations. These measures assess recovery of the
subject-specific random effects $b_i$. In particular, their behavior as
$n_i$ increases provides a finite-sample assessment of the increasing
within-subject regime considered in Proposition~\ref{Prop1}. Comparisons
with fixed $n_i$ and increasing $m$ are reported separately and are not
interpreted as direct evidence for the $n_i\rightarrow\infty$ condition in
Proposition~\ref{Prop1}. We assessed the predictor consistency result established in
Section~\ref{Consistency Predictor} using
\begin{align}
MSE(\widehat{\eta},\eta^0)
=
\frac{1}{N_{\mathrm{test}}}
\sum_{(i,j)\in\mathcal D_{\mathrm{test}}}
(\widehat{\eta}_{ij}-\eta_{ij}^{0})^2,
\label{eq:mse_eta}
\end{align}
where
\[
\eta_{ij}^{0}
=
f_0(X_{ij})+b_i
\]
denotes the true predictor and
\[
\widehat{\eta}_{ij}
=
f_{\widehat{\theta}}(X_{ij})+\widehat b_i
\]
denotes the fitted GNMM predictor for a testing observation. This measure
evaluates the discrepancy between the fitted GNMM predictor
$\widehat{\eta}_{ij}$ and the true predictor $\eta_{ij}^{0}$ considered in
Theorem~\ref{Thm2}. We assessed the conditional mean consistency result established in
Section~\ref{Consistency mean} using
\begin{align}
MSE(\widehat{\mu},\mu^0)
=
\frac{1}{N_{\mathrm{test}}}
\sum_{(i,j)\in\mathcal D_{\mathrm{test}}}
(\widehat{\mu}_{ij}-\mu_{ij}^{0})^2,
\label{eq:mse_mu}
\end{align}
where
\[
\mu_{ij}^{0}
=
g^{-1}(\eta_{ij}^{0}),
\qquad
\widehat{\mu}_{ij}
=
g^{-1}(\widehat{\eta}_{ij})
\]
denote the true and estimated conditional means, respectively. This measure
evaluates the discrepancy between the estimated conditional mean and the
true conditional mean considered in Corollary~\ref{Cor1}. For the continuous
response with the identity link,
$\widehat{\mu}_{ij}=\widehat{\eta}_{ij}$ and
$\mu_{ij}^{0}=\eta_{ij}^{0}$, so
$MSE(\widehat{\mu},\mu^0)$ and
$MSE(\widehat{\eta},\eta^0)$ coincide. For the binary response with the
logit link, these quantities are distinct. We evaluated out-of-sample predictive performance using
\begin{align}
MSPE_{\mathrm{GNMM}}
&=
\frac{1}{N_{\mathrm{test}}}
\sum_{(i,j)\in\mathcal D_{\mathrm{test}}}
(Y_{ij}-\widehat{\mu}_{ij})^2,
\label{eq:mspe}
\\
MAE
&=
\frac{1}{N_{\mathrm{test}}}
\sum_{(i,j)\in\mathcal D_{\mathrm{test}}}
|Y_{ij}-\widehat{\mu}_{ij}|,
\label{eq:mae}
\end{align}
where $Y_{ij}$ denotes the observed response for testing observation
$(i,j)$ and $\widehat{\mu}_{ij}$ is the corresponding conditional mean
predicted from the fitted GNMM. To directly evaluate the prediction-error consistency result established in
Section~\ref{Error consistency}, we also computed the oracle MSPE,
\begin{align}
MSPE_{\mathrm{Oracle}}
&=
\frac{1}{N_{\mathrm{test}}}
\sum_{(i,j)\in\mathcal D_{\mathrm{test}}}
(Y_{ij}-\mu_{ij}^{0})^2,
\label{eq:oracle_mspe}
\end{align}
and the excess MSPE,
\begin{align}
\Delta MSPE
&=
MSPE_{\mathrm{GNMM}}
-
MSPE_{\mathrm{Oracle}}.
\label{eq:excess_mspe}
\end{align}
Because the true conditional mean $\mu_{ij}^{0}$ is known in the simulation
study, $MSPE_{\mathrm{Oracle}}$ can be calculated directly. The quantity
$\Delta MSPE$ therefore provides a direct finite-sample assessment of the
excess prediction error considered in Theorem~\ref{Thm3}. The excess MSPE
was calculated within each simulation replication as the difference between
the GNMM and oracle MSPEs evaluated on the same test observations and was
then summarized across the $R=500$ replications. According to
Theorem~\ref{Thm3}, this difference is expected to approach zero as the
neural-network estimation error, model-class approximation error, and
random-effect estimation error decrease. Because $\Delta MSPE$ is an
empirical finite-test-sample quantity, small deviations around zero may occur
in individual simulation replications even though the corresponding
population excess mean squared prediction error is nonnegative. For the continuous response model in \eqref{eq:sim_cont}, the oracle
conditional mean is
\[
\mu_{ij}^{0}=f_0(X_{ij})+b_i,
\]
and hence the population oracle MSPE equals
$E(\epsilon_{ij}^{2})=1$. Therefore, the empirical
$MSPE_{\mathrm{Oracle}}$ is expected to be close to 1 and also provides a
useful numerical check on the simulation implementation. For the binary
response, the oracle conditional mean is
\[
\mu_{ij}^{0}
=
\frac{\exp(\eta_{ij}^{0})}
{1+\exp(\eta_{ij}^{0})},
\]
and the oracle MSPE depends on these generated conditional probabilities and
is computed directly from the simulated test observations using
$\mu_{ij}^{0}$. For binary responses, we additionally evaluated the area under the receiver
operating characteristic curve (AUC) \citep{hanley1982auc}, calculated from the observed binary
responses and the fitted conditional probabilities
$\widehat{\mu}_{ij}$. MAE and AUC were included as additional measures of
predictive performance and are not direct measures of the prediction-error
consistency result in Theorem~\ref{Thm3}.

The quantity $MSE(\widehat f,f_0)$ assesses recovery of the nonlinear
fixed-effect function and reflects both estimation and model-class
approximation error. The quantities
$\mathrm{cor}(\widehat b_i,b_i)$ and $MSE(\widehat b_i,b_i)$ assess the
recovery of the subject-specific random effects $b_i$. The quantity
$MSE(\widehat{\eta},\eta^0)$ evaluates the discrepancy between the fitted
predictor $\widehat{\eta}_{ij}$ and the true predictor $\eta_{ij}^{0}$,
while $MSE(\widehat{\mu},\mu^0)$ evaluates the discrepancy between the
estimated conditional mean $\widehat{\mu}_{ij}$ and the true conditional
mean $\mu_{ij}^{0}$. Finally, $\Delta MSPE$ directly assesses the excess
prediction error considered in Theorem~\ref{Thm3}, while MAE and AUC provide
additional measures of predictive performance. Together, these measures
provide finite-sample evidence corresponding to the theoretical properties
developed in Section~\ref{sec3}, without interpreting the finite-sample
simulation results as a proof of the asymptotic results.
\subsection{Simulation Results}

Table~\ref{tab:sim_continuous_recovery} summarizes the recovery and
estimation results for continuous responses. The results indicate that the
performance of GNMM improves as the amount of information available in the
data increases. In particular, the error in estimating the nonlinear
fixed-effect function, measured by $MSE(\widehat f,f_0)$, decreases from
0.372 to 0.077 as the simulation setting changes from
$(m,n_i)=(50,5)$ to $(200,20)$. This reduction indicates increasingly
accurate recovery of the underlying nonlinear fixed-effect function
$f_0(\cdot)$ and provides finite-sample evidence consistent with the
approximation and fixed-effect estimation properties established in
Theorem~\ref{Thm1} and Proposition~\ref{PropFixed}.

The quality of random-effect estimation also improves as the amount of
within-subject information increases. For fixed $m=50$, increasing $n_i$
from 5 to 20 decreases $MSE(\widehat b_i,b_i)$ from 0.161 to 0.064 and
increases $\mathrm{cor}(\widehat b_i,b_i)$ from 0.597 to 0.870. Similarly,
for $(m,n_i)=(200,20)$, the random-effect MSE is 0.060 and the correlation
between the estimated and true random effects reaches 0.873. These findings
show that the estimated random effects $\widehat b_i$ more closely recover
the true random effects $b_i$ as within-subject information increases and
are consistent with the increasing-$n_i$ random-effect consistency result
established in Proposition~\ref{Prop1}.

A similar pattern is observed for the fitted predictor and conditional mean.
The predictor error, measured by $MSE(\widehat{\eta},\eta^0)$, decreases
from 0.521 for $(m,n_i)=(50,5)$ to 0.134 for $(200,20)$. Because the
continuous-response model uses the identity link,
$MSE(\widehat{\mu},\mu^0)$ is identical to
$MSE(\widehat{\eta},\eta^0)$ in this setting. These reductions indicate that
the fitted GNMM predictor $\widehat{\eta}_{ij}$ and conditional mean
$\widehat{\mu}_{ij}$ approach their corresponding true values as the amount
of information increases, providing finite-sample evidence consistent with
Theorem~\ref{Thm2} and Corollary~\ref{Cor1}.
\begin{table}[htbp]
\centering
\scriptsize
\caption{Recovery and estimation results for continuous responses. Entries
are Monte Carlo means with Monte Carlo standard errors (MCSEs) in
parentheses based on $R=500$ replications.}
\label{tab:sim_continuous_recovery}

\begin{tabular}{ccccccc}
\hline
$m$ & $n_i$ &
$MSE(\widehat{f},f_0)$ &
$MSE(\widehat{b},b)$ &
$\mathrm{cor}(\widehat{b},b)$ &
$MSE(\widehat{\eta},\eta^0)$ &
$MSE(\widehat{\mu},\mu^0)$ \\
\hline
50  & 5  & 0.372 (0.006) & 0.161 (0.001) & 0.597 (0.004) &
0.521 (0.007) & 0.521 (0.007) \\
50  & 10 & 0.197 (0.003) & 0.106 (0.001) & 0.770 (0.003) &
0.292 (0.004) & 0.292 (0.004) \\
50  & 20 & 0.124 (0.003) & 0.064 (0.001) & 0.870 (0.002) &
0.178 (0.003) & 0.178 (0.003) \\
100 & 10 & 0.120 (0.003) & 0.098 (0.001) & 0.782 (0.002) &
0.213 (0.003) & 0.213 (0.003) \\
200 & 10 & 0.094 (0.004) & 0.097 (0.001) & 0.786 (0.001) &
0.188 (0.004) & 0.188 (0.004) \\
200 & 20 & 0.077 (0.003) & 0.060 (0.000) & 0.873 (0.001) &
0.134 (0.003) & 0.134 (0.003) \\
\hline
\end{tabular}
\end{table}

Table~\ref{tab:sim_continuous_prediction} presents the corresponding
prediction results for continuous responses. As the amount of information
increases, $MSPE_{\mathrm{GNMM}}$ decreases from 1.525 for
$(m,n_i)=(50,5)$ to 1.137 for $(200,20)$, while MAE decreases from
0.982 to 0.850. In contrast, $MSPE_{\mathrm{Oracle}}$ remains approximately
1 across all simulation settings, as expected under the Gaussian error
distribution with unit variance. Most importantly, the excess prediction
error $\Delta MSPE$ decreases from 0.523 to 0.135. This pattern indicates
that the prediction performance of GNMM moves toward that of the oracle as
the neural-network estimation error, approximation error, and random-effect
estimation error decrease. The results therefore provide direct
finite-sample evidence consistent with the prediction-error consistency
result established in Theorem~\ref{Thm3}.

\begin{table}[!htbp]
\centering
\caption{Prediction results for continuous responses. Entries are Monte
Carlo means with Monte Carlo standard errors (MCSEs) in parentheses based
on $R=500$ replications.}
\label{tab:sim_continuous_prediction}
\begin{tabular}{cccccc}
\hline
$m$ & $n_i$ &
$MSPE_{\mathrm{GNMM}}$ &
$MSPE_{\mathrm{Oracle}}$ &
$\Delta MSPE$ &
MAE \\
\hline
50  & 5  & 1.525 (0.012) & 1.002 (0.006) & 0.523 (0.010) & 0.982 (0.004) \\
50  & 10 & 1.294 (0.007) & 1.003 (0.005) & 0.290 (0.006) & 0.907 (0.003) \\
50  & 20 & 1.178 (0.006) & 0.999 (0.004) & 0.180 (0.004) & 0.866 (0.002) \\
100 & 10 & 1.214 (0.005) & 1.002 (0.003) & 0.212 (0.004) & 0.879 (0.002) \\
200 & 10 & 1.184 (0.005) & 0.996 (0.003) & 0.189 (0.004) & 0.868 (0.002) \\
200 & 20 & 1.137 (0.004) & 1.002 (0.002) & 0.135 (0.003) & 0.850 (0.001) \\
\hline
\end{tabular}
\end{table}

Table~\ref{tab:sim_binary_recovery} presents the recovery and estimation
results for binary responses. Similar overall improvements are observed as
the number of subjects $m$ and repeated observations per subject $n_i$
increase, although estimation is more difficult in the smallest binary
setting. For $(m,n_i)=(50,5)$, $MSE(\widehat f,f_0)$ and
$MSE(\widehat{\eta},\eta^0)$ are relatively large, indicating finite-sample
instability when only a small number of binary observations are available
within each subject. As the amount of information increases,
$MSE(\widehat f,f_0)$ decreases substantially, reaching 0.127 for
$(m,n_i)=(200,20)$. This pattern indicates increasingly accurate recovery
of the nonlinear fixed-effect structure and is consistent with the
approximation and fixed-effect estimation results in Theorem~\ref{Thm1} and
Proposition~\ref{PropFixed}.

The estimated random effects also become more accurate as within-subject
information increases. For fixed $m=50$, increasing $n_i$ from 5 to 20
decreases $MSE(\widehat b_i,b_i)$ from 0.219 to 0.156 and increases
$\mathrm{cor}(\widehat b_i,b_i)$ from 0.327 to 0.623. For
$(m,n_i)=(200,20)$, the correlation reaches 0.632. This pattern provides
finite-sample evidence consistent with Proposition~\ref{Prop1}, while also
illustrating the greater difficulty of recovering subject-specific random
effects from binary responses.

The predictor error also decreases substantially as the amount of
information increases. In particular, $MSE(\widehat{\eta},\eta^0)$
decreases from 6.826 in the smallest setting to 0.276 for
$(m,n_i)=(200,20)$. More importantly for probability prediction,
$MSE(\widehat{\mu},\mu^0)$ decreases from 0.054 to 0.011. Thus, although
the linear-predictor scale is relatively unstable in the smallest binary
setting, the estimated conditional probabilities become increasingly close
to the true probabilities as the sample information increases. These
findings provide finite-sample evidence consistent with the predictor
consistency result in Theorem~\ref{Thm2} and the conditional mean
consistency result in Corollary~\ref{Cor1}.
\begin{table}[htbp]
\centering
\scriptsize
\caption{Recovery and estimation results for binary responses. Entries are
Monte Carlo means with Monte Carlo standard errors (MCSEs) in parentheses
based on $R=500$ replications.}
\label{tab:sim_binary_recovery}

\begin{tabular}{ccccccc}
\hline
$m$ & $n_i$ &
$MSE(\widehat{f},f_0)$ &
$MSE(\widehat{b},b)$ &
$\mathrm{cor}(\widehat{b},b)$ &
$MSE(\widehat{\eta},\eta^0)$ &
$MSE(\widehat{\mu},\mu^0)$ \\
\hline
50  & 5  & 6.624 (0.711) & 0.219 (0.002) & 0.327 (0.005) &
6.826 (0.711) & 0.054 (0.001) \\
50  & 10 & 0.673 (0.041) & 0.192 (0.002) & 0.486 (0.005) &
0.847 (0.040) & 0.027 (0.000) \\
50  & 20 & 0.308 (0.005) & 0.156 (0.001) & 0.623 (0.004) &
0.452 (0.005) & 0.018 (0.000) \\
100 & 10 & 0.303 (0.005) & 0.189 (0.001) & 0.493 (0.003) &
0.489 (0.005) & 0.019 (0.000) \\
200 & 10 & 0.186 (0.003) & 0.189 (0.001) & 0.498 (0.002) &
0.373 (0.003) & 0.015 (0.000) \\
200 & 20 & 0.127 (0.003) & 0.151 (0.001) & 0.632 (0.002) &
0.276 (0.003) & 0.011 (0.000) \\
\hline
\end{tabular}

\end{table}
Table~\ref{tab:sim_binary_prediction} presents the prediction results for
binary responses. The $MSPE_{\mathrm{GNMM}}$ decreases from 0.253 for
$(m,n_i)=(50,5)$ to 0.210 for $(200,20)$, while the AUC increases from
0.615 to 0.697. The oracle MSPE remains stable at approximately 0.198
across the simulation settings. Most importantly, $\Delta MSPE$ decreases
from 0.054 to 0.011 as the amount of information increases. This reduction
shows that the excess prediction error of GNMM relative to the oracle
becomes progressively smaller, providing direct finite-sample evidence
consistent with Theorem~\ref{Thm3}. The MAE also decreases slightly from
0.428 to 0.417, while the increase in AUC indicates improved discrimination
for the binary response.

\begin{table}[!htbp]
\centering
\caption{Prediction results for binary responses. Entries are Monte Carlo
means with Monte Carlo standard errors (MCSEs) in parentheses based on
$R=500$ replications.}
\label{tab:sim_binary_prediction}
\begin{tabular}{ccccccc}
\hline
$m$ & $n_i$ &
$MSPE_{\mathrm{GNMM}}$ &
$MSPE_{\mathrm{Oracle}}$ &
$\Delta MSPE$ &
MAE & AUC \\
\hline
50  & 5  & 0.253 (0.001) & 0.198 (0.001) & 0.054 (0.001) &
0.428 (0.001) & 0.615 (0.003) \\
50  & 10 & 0.226 (0.001) & 0.199 (0.001) & 0.027 (0.001) &
0.424 (0.001) & 0.653 (0.002) \\
50  & 20 & 0.216 (0.001) & 0.198 (0.000) & 0.018 (0.000) &
0.419 (0.001) & 0.678 (0.001) \\
100 & 10 & 0.218 (0.000) & 0.198 (0.000) & 0.019 (0.000) &
0.423 (0.001) & 0.673 (0.002) \\
200 & 10 & 0.214 (0.000) & 0.198 (0.000) & 0.015 (0.000) &
0.421 (0.000) & 0.684 (0.001) \\
200 & 20 & 0.210 (0.000) & 0.198 (0.000) & 0.011 (0.000) &
0.417 (0.000) & 0.697 (0.001) \\
\hline
\end{tabular}
\end{table}

Overall, the simulation study provides finite-sample evidence consistent
with the theoretical results established in Section~\ref{sec3}. Recovery of the nonlinear fixed-effect function improves as the amount of information increases, providing finite-sample evidence for accurate
fixed-effect recovery and results consistent with the approximation and
fixed-effect estimation framework developed in Theorem~\ref{Thm1} and
Proposition~\ref{PropFixed}. The settings with fixed $m=50$ and increasing
$n_i$ show particularly clear improvement in random-effect recovery,
consistent with the increasing within-subject asymptotic regime underlying
Proposition~\ref{Prop1}. Likewise, the reductions in
$MSE(\widehat{\eta},\eta^0)$ and $MSE(\widehat{\mu},\mu^0)$ provide
finite-sample evidence consistent with predictor and conditional mean
consistency. Finally, the systematic reduction in $\Delta MSPE$ for both
continuous and binary responses shows that the prediction error of GNMM
moves toward the oracle prediction error as the available information
increases, consistent with Theorem~\ref{Thm3}. Collectively, these findings
show that the approximation, estimation, random-effect recovery, and
prediction properties developed theoretically are reflected in the
finite-sample behavior of GNMM.

The primary simulation summaries were computed using all replications that
produced finite parameter estimates and evaluation measures. Numerical
convergence was assessed separately using the convergence code returned by
the BFGS optimizer, with convergence code zero indicating that the formal
optimization convergence criterion was satisfied. Thus, a nonzero
convergence code does not necessarily indicate that a replication failed to
produce finite estimates; rather, it indicates that the optimizer did not
satisfy its formal convergence criterion within the specified optimization
settings. The optimization convergence rate ranged from 0.940 to 0.954
across the continuous-response settings and from 0.718 to 0.908 across the
binary-response settings. The lower convergence rates observed for the
smaller binary-response settings are consistent with greater numerical
difficulty in fitting the nonlinear mixed model when within-subject binary
information is limited. We report these convergence rates separately to
provide a transparent assessment of the numerical performance of the
estimation procedure.

\section{Real Data Analysis}\label{sec5}

The simulation study provides empirical support for the theoretical results
established in Section~\ref{sec3} under controlled settings where the true
data-generating mechanism is known. We next investigate the practical
performance of the proposed generalized neural network mixed model (GNMM)
using a real longitudinal dataset consisting of repeated measurements
collected from multiple subjects over time. Unlike the simulation study, the
true fixed-effect function $f_0(\cdot)$, the subject-specific random effects
$b_i$, the true predictor $\eta_{ij}^{0}$, and the true conditional mean
$\mu_{ij}^{0}$ are unknown in real data. Consequently, model performance is
evaluated through out-of-sample prediction rather than direct assessment of
the theoretical properties developed in Section~\ref{sec3}. The dataset was
divided into training and testing sets within subjects to preserve the
longitudinal structure of the data and evaluate prediction performance on
held-out observations from subjects represented in the training data. The
proposed GNMM was fitted using a feed-forward neural network for the
fixed-effect function together with subject-specific random intercepts.

\subsection{Dataset Description}

We analyzed a longitudinal schizophrenia dataset obtained from a mobile
health (mHealth) study of 15 individuals diagnosed with schizophrenia. The
dataset was collected from smartphone sensor and clinical symptom assessments
using the Beiwe digital phenotyping platform over an approximately 90-day
observation period \citep{torous2016} and was previously analyzed by
Mandel et al. \citep{mandel2023neural}. The dataset is available at
\url{https://onlinelibrary.wiley.com/doi/10.1111/biom.13615}. The dataset
contains repeated measurements of depression, anxiety, mobility,
communication, and behavioral characteristics collected from individuals.
The primary outcomes considered in this study were depression and anxiety
symptom severity, together with 15 predictors. The dataset has a longitudinal
structure with repeated observations collected from each subject over time.
The number of observations per individual varies substantially, ranging from
2 to 121 observations, with an average of approximately 53 observations per
subject. Consequently, the data exhibit both within-subject dependence and
subject-specific heterogeneity, making them particularly suitable for
evaluating mixed-effects methodologies. These characteristics also motivate
the use of the proposed GNMM, which combines flexible nonlinear fixed-effect
modeling with subject-specific random effects to accommodate the complex
longitudinal structure of the data.

For the continuous-response analysis, depression scores were analyzed as a
continuous outcome. For the binary-response analysis, depression scores were
converted into a binary outcome using the sample median as the cutoff. The
binary analysis was included to evaluate GNMM under a non-Gaussian response
and to provide a common binary outcome for comparing GNMM, GLMM, and ANN.
Because this dichotomization is used for methodological evaluation rather
than to define a clinically validated diagnostic threshold, the resulting
binary outcome should be interpreted as a data-defined classification
outcome.
For all models, the same predictor variables and training and testing
observations were used to ensure a comparable evaluation. The data were
divided into 70\% training and 30\% testing observations within each subject,
ensuring that all subjects in the testing set were also represented in the
training set. This design evaluates subject-specific prediction for existing
subjects, consistent with the random-effect prediction setting considered in
Section~\ref{sec3}. All predictive measures were calculated using the
held-out testing observations.

\subsection{Evaluation Criteria}

The theoretical results in Section~\ref{sec3} establish desirable properties
of the fitted predictor $\widehat{\eta}_{ij}$, the estimated conditional mean
$\widehat{\mu}_{ij}$, and the resulting prediction error. Since the true
predictor $\eta_{ij}^{0}$, conditional mean $\mu_{ij}^{0}$, and random effects
$b_i$ are unknown in real data, these theoretical quantities cannot be
evaluated directly. Therefore, model performance was assessed using
out-of-sample predictive measures that reflect the practical accuracy of the
fitted models. These real-data measures are used to evaluate practical
predictive performance and should not be interpreted as direct empirical
verification of the consistency results in Section~\ref{sec3}.

For the continuous depression outcome, predictive performance was evaluated
using the mean squared prediction error (MSPE) and mean absolute error (MAE),
defined in equations~\eqref{eq:mspe} and~\eqref{eq:mae}, respectively,
together with the root mean squared error (RMSE). Smaller values of these
measures indicate more accurate prediction of the observed responses. In
particular, MSPE provides the real-data counterpart of the squared-error
prediction criterion studied theoretically in
Subsection~\ref{Error consistency}; however, the oracle MSPE and excess MSPE
cannot be calculated because the true conditional mean is unknown.

For the binary depression outcome, predictive performance was evaluated
using the area under the receiver operating characteristic curve (AUC),
classification accuracy, sensitivity, and specificity. Larger values of
these measures indicate better discrimination or classification performance.
These measures provide practical assessments of predictive performance but
are not direct evaluations of the conditional mean consistency result in
Subsection~\ref{Consistency mean}, because the true conditional response
probabilities are unknown in the real dataset.

We compared the proposed GNMM with a generalized linear mixed model (GLMM)
and a conventional artificial neural network (ANN). The GLMM models
subject-specific heterogeneity through the random effects $b_i$ but assumes a
linear fixed-effect structure. In contrast, ANN models nonlinear fixed
effects but does not explicitly account for within-subject dependence through
subject-specific random effects. GNMM combines both modeling strategies by
using a neural network to represent potentially nonlinear fixed-effect
relationships while accounting for subject-specific heterogeneity through
the random effects $b_i$. All three models were fitted using the same
training observations and predictor variables and were evaluated on the same
testing observations to provide a fair comparison.
\subsection{Continuous Depression Analysis}

Table~\ref{tab:real_continuous} summarizes the prediction performance for
the continuous depression outcome. Among the three competing methods, GLMM
achieved the smallest prediction error, followed by GNMM and ANN. Specifically,
GLMM produced an MSPE of 0.041, RMSE of 0.203, and MAE of 0.133. GNMM
achieved an MSPE of 0.178, RMSE of 0.422, and MAE of 0.269, whereas ANN
resulted in the largest prediction errors across all three measures. The
better performance of GLMM suggests that, for this dataset and outcome, a
substantial portion of the predictive structure may be adequately captured
by the subject-specific random effects together with a linear fixed-effect
component. Nevertheless, GNMM consistently outperformed ANN, indicating that
explicitly accounting for subject-specific heterogeneity and within-subject
dependence can improve prediction relative to a conventional neural network
that does not include random effects.

Although the theoretical results in Section~\ref{sec3} do not imply that
GNMM will outperform every competing method in finite samples, the results
show that GNMM provides useful out-of-sample predictive performance while
simultaneously accommodating nonlinear fixed effects and subject-specific
heterogeneity. The superior performance of GLMM also indicates that the
additional nonlinear flexibility of GNMM does not necessarily improve
prediction when a simpler mixed-effects structure is sufficient for a given
outcome. Since the true predictor $\eta_{ij}^{0}$ and conditional mean
$\mu_{ij}^{0}$ are unknown in the real dataset, these results should be
interpreted as evidence of practical predictive performance rather than
direct verification of the consistency results in
Theorems~\ref{Thm2} and~\ref{Thm3}.

\begin{table}[htbp]
\centering
\caption{Prediction performance for the continuous depression outcome.
MSPE and MAE are defined in Equations~\eqref{eq:mspe} and~\eqref{eq:mae},
respectively, while RMSE denotes the root mean squared error. Smaller values
indicate better out-of-sample predictive performance.}
\label{tab:real_continuous}
\begin{tabular}{lccc}
\hline
Method & MSPE & RMSE & MAE \\
\hline
GLMM & 0.041 & 0.203 & 0.133 \\
ANN  & 0.297 & 0.545 & 0.337 \\
GNMM & 0.178 & 0.422 & 0.269 \\
\hline
\end{tabular}
\end{table}

\subsection{Binary Depression Analysis}

Table~\ref{tab:real_binary} presents the classification results for the
binary depression outcome. ANN achieved the highest overall classification
performance with an AUC of 0.926 and an accuracy of 0.857. GNMM achieved an
AUC of 0.893 and an accuracy of 0.825, whereas GLMM achieved an AUC of 0.841
and an accuracy of 0.833. Compared with GLMM, the higher AUC obtained by
GNMM suggests that allowing a flexible nonlinear fixed-effect component may
improve discrimination for this outcome. At the same time, the inclusion of
the random effects $b_i$ enables GNMM to account explicitly for
subject-specific heterogeneity and within-subject dependence, which are not
represented by the conventional ANN.

Although ANN achieved the highest AUC for this dataset, GNMM remained
competitive while improving upon GLMM in discrimination ability. These
results illustrate the practical trade-off between nonlinear modeling
flexibility and explicit modeling of subject-specific dependence. In
particular, ANN may benefit from greater nonlinear flexibility when
classification is the primary objective, whereas GNMM additionally models
the longitudinal dependence structure through the random effects. Because
the true fixed-effect function $f_0(\cdot)$, true predictor
$\eta_{ij}^{0}$, and true conditional response probabilities
$\mu_{ij}^{0}$ are unknown in the real data, the observed classification
performance does not directly verify the approximation or consistency
results in Section~\ref{sec3}. Instead, it demonstrates the practical
predictive performance of the competing models under the observed
longitudinal structure.

\begin{table}[htbp]
\centering
\caption{Prediction performance for the binary depression outcome. AUC
denotes the area under the receiver operating characteristic curve, while
accuracy, sensitivity, and specificity summarize classification performance
based on the estimated conditional probabilities $\widehat{\mu}_{ij}$.
Larger values indicate better predictive performance.}
\label{tab:real_binary}
\begin{tabular}{lcccc}
\hline
Method & AUC & Accuracy & Sensitivity & Specificity \\
\hline
GLMM & 0.841 & 0.833 & 0.787 & 0.861 \\
ANN  & 0.926 & 0.857 & 0.787 & 0.899 \\
GNMM & 0.893 & 0.825 & 0.745 & 0.873 \\
\hline
\end{tabular}
\end{table}

\section{Discussion and Conclusion}\label{sec6}

In this paper, we establish a theoretical framework for generalized neural
network mixed models (GNMMs) for longitudinal and clustered data. We prove
the approximation property of the neural-network fixed-effect component and
establish fixed-effect estimation consistency within a given neural-network
class, together with random-effect consistency under increasing within-subject
information, predictor consistency, conditional mean consistency, and
prediction-error consistency relative to an oracle predictor. These results
provide a statistical justification for using GNMM to model complex nonlinear
relationships while accounting for subject-specific heterogeneity. The
simulation study provides finite-sample evidence corresponding to these
properties. As the numbers of subjects and repeated observations increase,
the nonlinear fixed-effect function is recovered more accurately, the
estimated random effects and fitted predictor approach their true
counterparts, and the conditional mean estimation error decreases. In
addition, the excess mean squared prediction error relative to the oracle
decreases as the amount of information increases. The real-data analysis
further demonstrates the practical applicability of GNMM using a longitudinal
schizophrenia dataset. GNMM achieved competitive predictive performance for
both continuous and binary depression outcomes while accounting for nonlinear
fixed effects and within-subject heterogeneity.

An important direction for further development is to extend the theoretical
framework of GNMM to more complex deep learning architectures. In particular,
integrating generative adversarial networks (GANs) with mixed-effects
structures may provide a flexible approach for modeling and generating complex
longitudinal and clustered data. Recent theoretical developments in GANs provide motivation for investigating the statistical
properties of such extensions.

\backmatter
\bmhead{Supplementary information}

The supplementary material contains additional simulation results, implementation details, and the R code used for the simulation study and real-data analysis.

\section*{Declarations}

\textbf{Funding}

The authors received no specific funding for this work.

\textbf{Conflict of interest}

The authors declare that they have no competing interests.

\textbf{Ethics approval and consent to participate}

Not applicable.

\textbf{Consent for publication}

Not applicable.

\textbf{Data availability}
 The processed data and simulation results are available from the corresponding author upon reasonable request.

\textbf{Materials availability}

Not applicable.

\textbf{Code availability}

The R code used to  conduct the simulation study, and perform the real-data analysis is available from the corresponding author upon reasonable request.
\newpage
\textbf{Author Contributions.} 
Mahmud Hasan conceived the study, developed the framework, designed the simulation study, conducted the real-data analysis, and drafted the manuscript. Md. Moinul Ahsan contributed to the R code and the generation of the simulation and real-data results using cluster computing. Mathias Nthiani Muia contributed to the theoretical development and the mathematical proofs. Mous-Abou Hamadou contributed to the interpretation of the results and refinement of the theoretical proofs. Chathurika S. Abeykoon critically reviewed the manuscript and provided valuable suggestions that improved its quality. All authors reviewed, revised, and approved the final version of the manuscript.

\begin{appendices}

\begin{proof}[Proof of Theorem~\ref{Thm1}]
Since \(f_0\) is continuous on the compact set \(\mathcal X\), the universal approximation property implies that for every \(\varepsilon>0\), there exists \(f_{\theta_K}\in\mathcal F_K\) such that
\[
\sup_{x\in\mathcal X}
\left|
f_{\theta_K}(x)-f_0(x)
\right|
<
\varepsilon.
\]
From \eqref{Tpr} and \eqref{GPr},
\[
\eta_{ij}^{K}
=
f_{\theta_K}(X_{ij})
+
Z_{ij}^{T}b_i
\]
and
\[
\eta_{ij}^{0}
=
f_0(X_{ij})
+
Z_{ij}^{T}b_i.
\]

Then,
\[
\begin{aligned}
\left|
\eta_{ij}^{K}
-
\eta_{ij}^{0}
\right|
&=
\left|
\left\{
f_{\theta_K}(X_{ij})
+
Z_{ij}^{T}b_i
\right\}
-
\left\{
f_0(X_{ij})
+
Z_{ij}^{T}b_i
\right\}
\right| \\
&=
\left|
f_{\theta_K}(X_{ij})
-
f_0(X_{ij})
\right|.
\end{aligned}
\]

Since \(X_{ij}\in\mathcal X\), it follows that
\[
\left|
f_{\theta_K}(X_{ij})
-
f_0(X_{ij})
\right|
\le
\sup_{x\in\mathcal X}
\left|
f_{\theta_K}(x)-f_0(x)
\right|.
\]

Therefore,
\[
\left|
\eta_{ij}^{K}
-
\eta_{ij}^{0}
\right|
<
\varepsilon.
\]

Taking the supremum over all \(i\) and \(j\), we obtain
\[
\sup_{i,j}
\left|
\eta_{ij}^{K}
-
\eta_{ij}^{0}
\right|
<
\varepsilon.
\]
\end{proof}

\begin{proof}[Proof of Proposition~\ref{PropFixed}]

Let \(Q_{N,K}(\theta)\) denote the sample quasi-likelihood criterion for the
fixed network class \(\mathcal F_K\), and let \(Q_K(\theta)\) denote its
population counterpart. By definition,

\[
\widehat{\theta}
=
\arg\max_{\theta\in\Theta_K}
Q_{N,K}(\theta),
\]

while, by the identification assumption,

\[
\theta_K^{*}
=
\arg\max_{\theta\in\Theta_K}
Q_K(\theta).
\]

By assumption,

\[
\sup_{\theta\in\Theta_K}
\left|
Q_{N,K}(\theta)
-
Q_K(\theta)
\right|
\xrightarrow{p}
0.
\]

Since \(\Theta_K\) is compact, \(Q_K(\theta)\) is continuous on
\(\Theta_K\), and \(Q_K(\theta)\) has the unique maximizer
\(\theta_K^{*}\), the argmax consistency theorem
\citep{vandervaart1998,newey1994} implies that

\[
\widehat{\theta}
\xrightarrow{p}
\theta_K^{*}.
\]

It remains to establish convergence of the corresponding neural-network
functions. By assumption, the mapping

\[
\theta
\longmapsto
f_{\theta}
\]

is continuous with respect to the uniform norm on \(\mathcal X\). Therefore,
for every \(\varepsilon>0\), there exists \(\delta>0\) such that

\[
\|\theta-\theta_K^{*}\|<\delta
\]

implies

\[
\sup_{x\in\mathcal X}
\left|
f_{\theta}(x)
-
f_{\theta_K^{*}}(x)
\right|
<
\varepsilon.
\]

Hence,

\[
\begin{aligned}
P\left(
\sup_{x\in\mathcal X}
\left|
f_{\widehat{\theta}}(x)
-
f_{\theta_K^{*}}(x)
\right|
>
\varepsilon
\right)
&\leq
P\left(
\|\widehat{\theta}-\theta_K^{*}\|
\geq
\delta
\right).
\end{aligned}
\]

Since

\[
\widehat{\theta}
\xrightarrow{p}
\theta_K^{*},
\]

the probability on the right-hand side converges to zero. Therefore,

\[
\sup_{x\in\mathcal X}
\left|
f_{\widehat{\theta}}(x)
-
f_{\theta_K^{*}}(x)
\right|
\xrightarrow{p}
0,
\]

or equivalently,

\[
\sup_{x\in\mathcal X}
\left|
f_{\widehat{\theta}}(x)
-
f_{\theta_K^{*}}(x)
\right|
=
o_p(1).
\]

Thus, within the fixed neural-network class \(\mathcal F_K\), the fitted
neural-network function consistently estimates its population target
\(f_{\theta_K^{*}}\).
\end{proof}

\begin{proof}[Proof of Proposition \ref{Prop1}]

Let \(\ell_i(b)\) denote the subject-specific conditional log-likelihood for subject \(i\), written as a function of the random effect \(b\). Let \(\widehat{\ell}_i(b)\) denote the corresponding fitted log-likelihood obtained after replacing \(f_0\) by \(f_{\widehat{\theta}}\). The estimated random effect is defined by

\[
\widehat b_i
=
\arg\max_{b\in\mathcal B}
\widehat{\ell}_i(b),
\]

where \(\mathcal B\) denotes the parameter space for \(b_i\).

Define

\[
\eta_{ij}^{0}(b)
=
f_0(X_{ij})
+
Z_{ij}^{T}b
\]

and

\[
\widehat{\eta}_{ij}(b)
=
f_{\widehat{\theta}}(X_{ij})
+
Z_{ij}^{T}b.
\]

By the assumed consistency of \(f_{\widehat{\theta}}\),

\[
\sup_{x\in\mathcal X}
\left|
f_{\widehat{\theta}}(x)-f_0(x)
\right|
=
o_p(1).
\]

For any \(b\in\mathcal B\), we have

\[
\begin{aligned}
\left|
\widehat{\eta}_{ij}(b)
-
\eta_{ij}^{0}(b)
\right|
&=
\left|
\left\{
f_{\widehat{\theta}}(X_{ij})
+
Z_{ij}^{T}b
\right\}
-
\left\{
f_0(X_{ij})
+
Z_{ij}^{T}b
\right\}
\right| \\
&=
\left|
f_{\widehat{\theta}}(X_{ij})
-
f_0(X_{ij})
\right| \\
&\le
\sup_{x\in\mathcal X}
\left|
f_{\widehat{\theta}}(x)-f_0(x)
\right| \\
&=
o_p(1).
\end{aligned}
\]

Hence,

\[
\sup_{b\in\mathcal B}
\left|
\widehat{\eta}_{ij}(b)
-
\eta_{ij}^{0}(b)
\right|
=
o_p(1).
\]

Since the conditional log-likelihood is a continuous function of the linear predictor, the difference between the fitted and true log-likelihoods also converges to zero after normalization. Therefore,

\[
\sup_{b\in\mathcal B}
\left|
\frac{1}{n_i}\widehat{\ell}_i(b)
-
\frac{1}{n_i}\ell_i(b)
\right|
\xrightarrow{p}
0.
\]

By the identifiability assumption, the true random effect \(b_i\) uniquely maximizes the limiting objective function, that is,

\[
b_i
=
\arg\max_{b\in\mathcal B}
\ell_i(b).
\]

Since \(\widehat{\ell}_i(b)/n_i\) converges uniformly to \(\ell_i(b)/n_i\), the argmax continuous mapping theorem \citep{vandervaart1998} implies that

\[
\widehat b_i
=
\arg\max_{b\in\mathcal B}
\widehat{\ell}_i(b)
\xrightarrow{p}
\arg\max_{b\in\mathcal B}
\ell_i(b)
=
b_i.
\]

Thus,

\[
\widehat b_i
\xrightarrow{p}
b_i.
\]
\end{proof}

\begin{proof}[Proof of Theorem \ref{Thm2}]

From \eqref{Tpr} and \eqref{GPr},

\[
\eta_{ij}^{0}
=
f_{0}(X_{ij})
+
Z_{ij}^{T}b_i,
\qquad
\widehat{\eta}_{ij}
=
f_{\widehat{\theta}}(X_{ij})
+
Z_{ij}^{T}\widehat b_i.
\]

Therefore,

\[
\begin{aligned}
\widehat{\eta}_{ij}
-
\eta_{ij}^{0}
&=
\left\{
f_{\widehat{\theta}}(X_{ij})
+
Z_{ij}^{T}\widehat b_i
\right\}
-
\left\{
f_{0}(X_{ij})
+
Z_{ij}^{T}b_i
\right\} \\
&=
f_{\widehat{\theta}}(X_{ij})
-
f_{0}(X_{ij})
+
Z_{ij}^{T}
(\widehat b_i-b_i).
\end{aligned}
\]

Taking absolute values and applying the triangle inequality,

\[
\begin{aligned}
\left|
\widehat{\eta}_{ij}
-
\eta_{ij}^{0}
\right|
&\le
\left|
f_{\widehat{\theta}}(X_{ij})
-
f_{0}(X_{ij})
\right|
+
\left|
Z_{ij}^{T}
(\widehat b_i-b_i)
\right|.
\end{aligned}
\]

Since \(\mathcal{F}_{K}\) satisfies the universal approximation property and \(f_{\widehat{\theta}}\) is a consistent estimator of \(f_{0}\),

\[
\sup_{x\in\mathcal X}
\left|
f_{\widehat{\theta}}(x)
-
f_{0}(x)
\right|
=
o_p(1).
\]

Hence,

\[
\left|
f_{\widehat{\theta}}(X_{ij})
-
f_{0}(X_{ij})
\right|
\le
\sup_{x\in\mathcal X}
\left|
f_{\widehat{\theta}}(x)
-
f_{0}(x)
\right|
=
o_p(1).
\]

For the second term, by the Cauchy--Schwarz inequality,

\[
\begin{aligned}
\left|
Z_{ij}^{T}
(\widehat b_i-b_i)
\right|
&\le
\|Z_{ij}\|
\,
\|\widehat b_i-b_i\| \\
&\le
C_Z
\,
\|\widehat b_i-b_i\|.
\end{aligned}
\]

By Proposition~\ref{Prop1},

\[
\widehat b_i
\xrightarrow{p}
b_i,
\]

which implies

\[
\|\widehat b_i-b_i\|
=
o_p(1).
\]

Therefore,

\[
\left|
Z_{ij}^{T}
(\widehat b_i-b_i)
\right|
=
o_p(1).
\]

Combining the above results,

\[
\begin{aligned}
\left|
\widehat{\eta}_{ij}
-
\eta_{ij}^{0}
\right|
&\le
o_p(1)
+
o_p(1) \\
&=
o_p(1).
\end{aligned}
\]

Thus,

\[
\widehat{\eta}_{ij}
-
\eta_{ij}^{0}
=
o_p(1),
\]

or equivalently,

\[
\widehat{\eta}_{ij}
\xrightarrow{p}
\eta_{ij}^{0}.
\]

This completes the proof.

\end{proof}
\begin{proof}[Proof of Corollary~\ref{Cor1}]

From Theorem~\ref{Thm2} in \eqref{CP}, the estimated predictor satisfies

\[
\widehat{\eta}_{ij}
\xrightarrow{p}
\eta_{ij}^{0}.
\]

Since the inverse link function \(g^{-1}\) is continuous at
\(\eta_{ij}^{0}\), the Continuous Mapping Theorem
\citep{vandervaart1998} implies that

\[
g^{-1}(\widehat{\eta}_{ij})
\xrightarrow{p}
g^{-1}(\eta_{ij}^{0}).
\]

Substituting the definitions of
\(\widehat{\mu}_{ij}\) and \(\mu_{ij}^{0}\), we obtain

\[
\widehat{\mu}_{ij}
\xrightarrow{p}
\mu_{ij}^{0}.
\]

Therefore,

\[
\widehat{\mu}_{ij}
-
\mu_{ij}^{0}
=
o_p(1),
\]

or equivalently,

\[
|\widehat{\mu}_{ij}-\mu_{ij}^{0}|
=
o_p(1).
\]

\end{proof}
\begin{proof}[Proof of Theorem~\ref{Thm3}]

From the definitions in \eqref{eq:MSPE_GNMM} and
\eqref{eq:MSPE_Oracle},

\[
\begin{aligned}
MSPE_{\mathrm{GNMM}}
-
MSPE_{\mathrm{Oracle}}
&=
E\left[
\left(
Y_{ij}-\widehat{\mu}_{ij}
\right)^2
-
\left(
Y_{ij}-\mu_{ij}^{0}
\right)^2
\right].
\end{aligned}
\]

Since
\[
\mu_{ij}^{0}
=
E(Y_{ij}\mid X_{ij},Z_{ij},b_i),
\]
we have
\[
E\left[
Y_{ij}-\mu_{ij}^{0}
\mid
X_{ij},Z_{ij},b_i
\right]
=
0.
\]

Using
\[
Y_{ij}-\widehat{\mu}_{ij}
=
(Y_{ij}-\mu_{ij}^{0})
-
(\widehat{\mu}_{ij}-\mu_{ij}^{0}),
\]
and taking expectations, the cross-product term vanishes. Therefore,

\[
MSPE_{\mathrm{GNMM}}
-
MSPE_{\mathrm{Oracle}}
=
E\left[
\left(
\widehat{\mu}_{ij}
-
\mu_{ij}^{0}
\right)^2
\right]
\geq 0.
\]

Since \(g^{-1}\) is Lipschitz continuous with constant \(L_g\),

\[
\left|
\widehat{\mu}_{ij}
-
\mu_{ij}^{0}
\right|
\leq
L_g
\left|
\widehat{\eta}_{ij}
-
\eta_{ij}^{0}
\right|.
\]

From the predictor-error decomposition in Theorem~\ref{Thm2},

\[
\begin{aligned}
\widehat{\eta}_{ij}-\eta_{ij}^{0}
&=
\left[
f_{\widehat{\theta}}(X_{ij})
-
f_{\theta_K^{*}}(X_{ij})
\right]
+
\left[
f_{\theta_K^{*}}(X_{ij})
-
f_0(X_{ij})
\right]
\\
&\quad+
Z_{ij}^{T}(\widehat b_i-b_i).
\end{aligned}
\]

Hence, using
\(\sup_{i,j}\|Z_{ij}\|\leq C_Z\),

\[
\begin{aligned}
\left|
\widehat{\eta}_{ij}
-
\eta_{ij}^{0}
\right|
&\leq
\sup_{x\in\mathcal X}
\left|
f_{\widehat{\theta}}(x)
-
f_{\theta_K^{*}}(x)
\right|
\\
&\quad+
\sup_{x\in\mathcal X}
\left|
f_{\theta_K^{*}}(x)
-
f_0(x)
\right|
+
C_Z\|\widehat b_i-b_i\|.
\end{aligned}
\]

Using the inequality
\[
(a+b+c)^2
\leq
3(a^2+b^2+c^2),
\]
we obtain

\[
\begin{aligned}
\left(
\widehat{\mu}_{ij}
-
\mu_{ij}^{0}
\right)^2
&\leq
3L_g^2
\Bigg\{
\sup_{x\in\mathcal X}
\left|
f_{\widehat{\theta}}(x)
-
f_{\theta_K^{*}}(x)
\right|^2
\\
&\qquad+
\sup_{x\in\mathcal X}
\left|
f_{\theta_K^{*}}(x)
-
f_0(x)
\right|^2
+
C_Z^2
\|\widehat b_i-b_i\|^2
\Bigg\}.
\end{aligned}
\]

Taking expectations on both sides gives

\[
\begin{aligned}
0
&\leq
MSPE_{\mathrm{GNMM}}
-
MSPE_{\mathrm{Oracle}}
\\
&\leq
3L_g^2
\Bigg\{
E\left[
\sup_{x\in\mathcal X}
\left|
f_{\widehat{\theta}}(x)
-
f_{\theta_K^{*}}(x)
\right|^2
\right]
\\
&\qquad+
\sup_{x\in\mathcal X}
\left|
f_{\theta_K^{*}}(x)
-
f_0(x)
\right|^2
+
C_Z^2
E\left[
\|\widehat b_i-b_i\|^2
\right]
\Bigg\}.
\end{aligned}
\]

Under the conditions of Theorem~\ref{Thm3}, each term on the right-hand
side converges to zero. Therefore,

\[
MSPE_{\mathrm{GNMM}}
-
MSPE_{\mathrm{Oracle}}
\longrightarrow 0.
\]

This completes the proof.

\end{proof}



\end{appendices}


\end{document}